\pdfoutput=1
\documentclass[10pt,conference,twocolumn]{IEEEtran}
\usepackage{array}
\usepackage{amsmath,graphicx,epsfig,float,subfigure,times,latexsym,verbatim,mathrsfs,amssymb,textcomp,txfonts,color}
\usepackage[flushleft]{threeparttable}
\usepackage{algorithm}
\usepackage{algorithmic}
\usepackage{url}
\usepackage{bm}
\newtheorem{proposition}{Proposition}

\renewcommand{\baselinestretch}{1}
\title{M$^2$I Communication: From Theoretical Modeling to Practical Design
\vspace{-15pt}\\}

\author{%
Hongzhi Guo~and~ Zhi~Sun
\thanks{This work was supported by the US National Science Foundation (NSF) under Grant No. 1547908.}
\vspace{5pt}\\
State University of New York at Buffalo, Buffalo, NY, USA \\
\{hongzhig, zhisun\}@buffalo.edu%
}
\IEEEoverridecommandlockouts
\begin{document}
\maketitle

%
\begin{abstract}
Wireless communications in complex environments are constrained by lossy media and complicated structures. Magnetic Induction (MI) has been proved to be an efficient solution to extend the communication range. Due to the small coil antenna's physical limitation, however, MI's communication range is still very limited. To this end, Metamaterial-enhanced Magnetic Induction (M$^2$I) communication has been proposed and the theoretical results suggest that it can significantly increase the communication performance, namely, data rate and communication range. Nevertheless, currently, the real implementation of M$^2$I is still a challenge and there is no guideline on design and fabrication of spherical metamaterial. In this paper, we propose a practical design by using a spherical coil array to realize M$^2$I and we prove that it can achieve negative permeability and there exists a resonance condition where the radiated magnetic field can be significantly amplified. The radiation and communication performance are evaluated and full-wave simulation in COMSOL Multiphysics is conducted to validate the design objectives. By using the spherical coil array-based M$^2$I, the communication range can be significantly extended, exactly as we predicted in the theoretical model.
\end{abstract}

%


\section{Introduction}
Although terrestrial wireless communication has been well developed and extensively utilized, its counterpart in complex environments, such as underground, oil reservoirs and nuclear plants, is still in its infancy. Wireless communication in such environments can enable a large number of important applications, such as harsh environment monitoring by using wireless sensors, miners rescue, and mitigation in nuclear plants using robots. Wireless signals in complex environments suffer from high absorption and multiple scattering which results in significant propagation loss. Therefore, extremely large antenna and high transmission power have to be used to accomplish signal transmissions. However, the devices for wireless communication is becoming smaller and smaller, such as sensors for wireless sensor networks and portable devices for wireless area networks. These devices can not be equipped with such large antennas or provide enough power. As a result, a new communication mechanism is highly desired to solve this problem.

Magnetic Induction (MI) as a promising solution has been envisioned to enable long-distance wireless communications in underground~\cite{Sun_MI_TAP_2010} and underwater~\cite{Guo_underwater}. Wireless signals are transmitted by using the reactive power constrained in the near field of a coil rather than using the propagating wave. Thus, it suffers from lower propagation loss and signal delay. Also, it enjoys a stable channel since the permeability is the same for most of the natural materials. Even the generated power by a coil antenna attenuate slower than that by an electrical antenna, the intrinsic physical limitation of the highly inductive coil antenna prevents further improving its efficiency~\cite{karlsson2004physical}. Metamaterial is introduced to MI in~\cite{guo2015m2i}, where an ideal negative-permeability metamaterial shell is utilized to surround a loop antenna, as shown in Fig.~\ref{fig:review1}. The near-field reactive power generated by the loop antenna can be matched by a metamaterial shell. The theoretical results predict that a pocket-size loop antenna can achieve around 20~m communication range with high data rate. However, since the considered metematerial in~\cite{guo2015m2i} is ideally homogeneous and isotropic which is not available in reality, a practical design is highly desired to validate the theoretically predicted results.

Metamaterial is composed of periodic artificial metallic or dielectric atoms. It can demonstrate negative permeability and permittivity, which can manipulate electromagnetic waves in extraordinary ways~\cite{pendry1999magnetism}. The periodic metamaterial components can be organized in many ways. In~\cite{Wang_WPT,scarborough2012experimental}, metamaterial unit cells formed a slab which is utilized for high-efficiency wireless power transmission. In \cite{xie2012proposal}, a metamaterial cylindrical shell is fabricated to improve the accuracy of Magnetic Resonance Imaging (MRI). More relevantly, most of existing works on metamaterial cloaking at GHz or THz bands, which makes the objects inside a spherical shell invisible, consider the ideally homogeneous metamaterial. Nevertheless, when it comes to implementation, spherical cloaks are simplified to cylinders due to the complexity of the 3D structure~\cite{schurig2006metamaterial}. Currently, it is still a great challenge to make metamaterial spherical. Besides the shape, since metamaterial is a kind of effective media~\cite{EMT_theory}, the effective parameters, such as effective permeability and thickness, are hard to extract. Since in~\cite{guo2015m2i}, the high efficiency of M$^2$I strongly depends on the negative permeability and thickness, it is crucial to find them out.

In this paper, the ideally theoretical model of Metamaterial-enhanced Magnetic Induction (M$^2$I) communication in~\cite{guo2015m2i} is pushed forward to a practical design. Specifically, we substitute the ideally homogeneous and isotropic metamaterial with a spherical coil array to realize M$^2$I. A large number of small coils are uniformly placed on a spherical shell to enhance the radiated field by a loop antenna, which in turn increases the communication range and data rate. We prove that this spherical coil array can achieve negative permeability. In addition, the optimal configuration of the proposed spherical coil array is found and its communication performances are similar as the ideal M$^2$I in~\cite{guo2015m2i}, both of which are much better than the original MI~\cite{Sun_MI_TAP_2010}. The results are evaluated and validated by full-wave simulations. Generally, the contribution of this paper, which has not been reported in existing works, can be summarized as follows: 1) a practical design based on a spherical coil array is proposed to realize the ideal M$^2$I; 2) we find the optimal conditions to achieve the negative permeability and the communication range is significantly increased; 3) we provide both insightful design guidelines and validations by using full-wave simulations in COMSOL Multiphysics.
\begin{figure*}[t]
\begin{minipage}{0.32\textwidth}
\centering
  \includegraphics[width=0.9\textwidth]{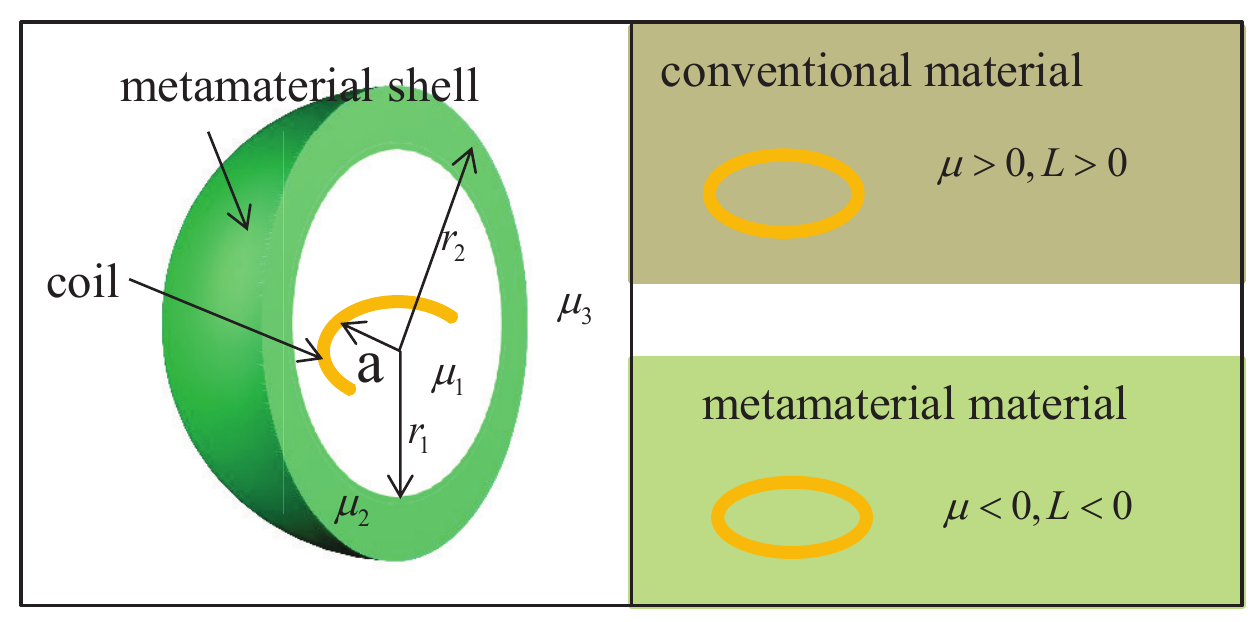}
  \caption{Illustration of M$^2$I.}
    \label{fig:review1}
\end{minipage}\quad
\begin{minipage}{0.2\textwidth}
\centering
  \includegraphics[width=0.95\textwidth]{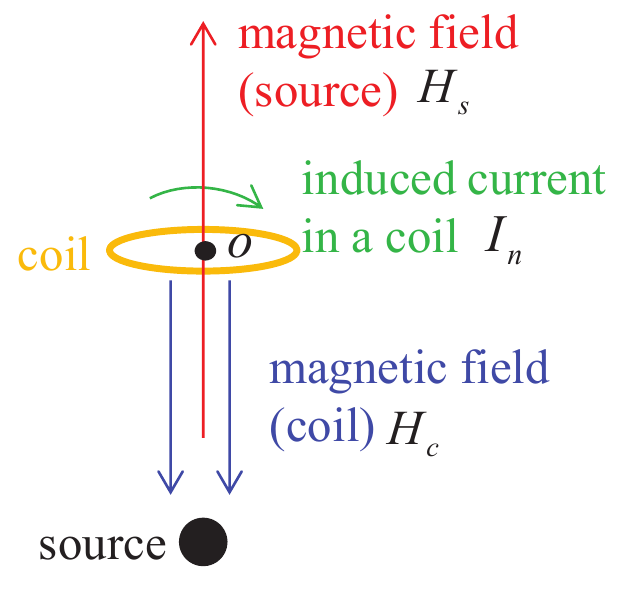}
  \caption{Illustration of effective $\mu$.}
    \label{fig:ideal2}
\end{minipage}\quad
\begin{minipage}{0.43\textwidth}
\centering
  \includegraphics[width=0.9\textwidth]{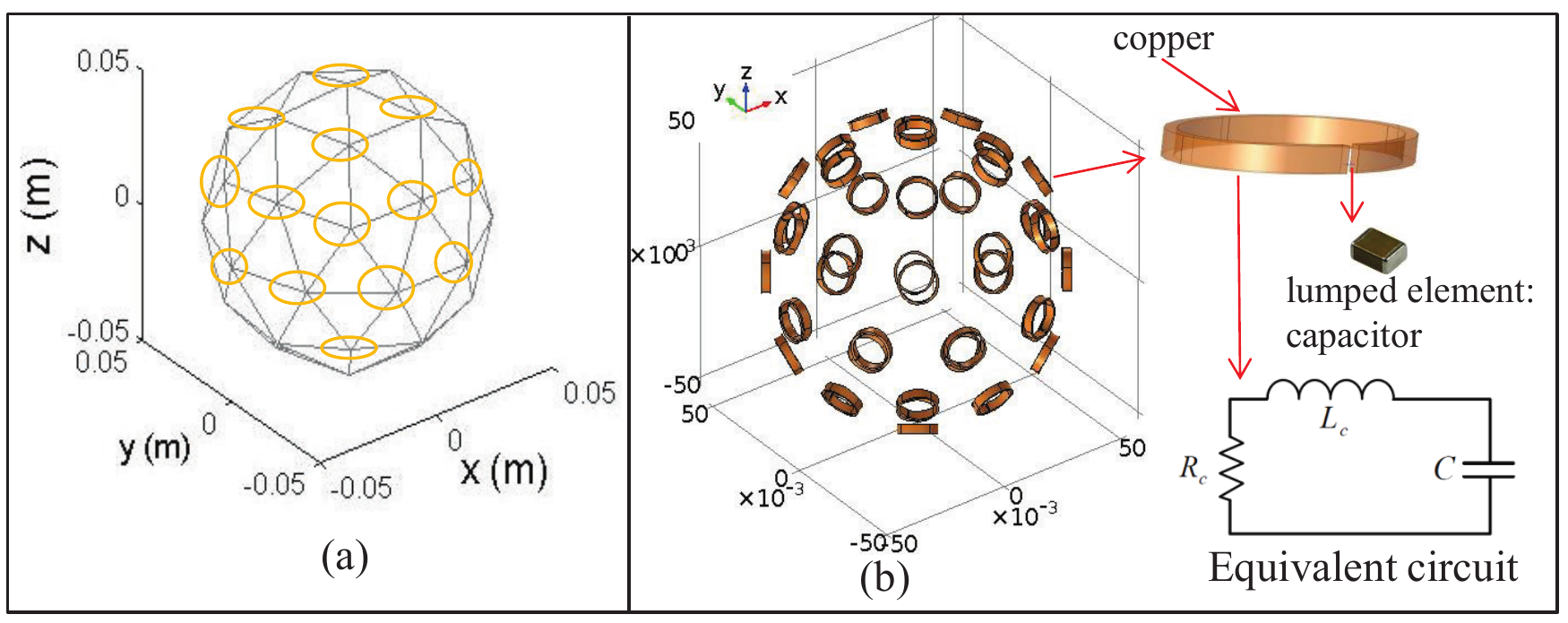}
  \caption{Spherical coil array. (a) Pentakis icosidodecahedron with 42 vertices; each vertex is the center of a coil; the orientation of a coil is the radial direction from the center of the sphere to the corresponding vertex. (b) 3D model and equivalent circuit.}
    \label{fig:geo1}
\end{minipage}\vspace{-13pt}
\end{figure*}

The following of this paper is organized as follows. First, the theoretical details of the ideal M$^2$I is reviewed in Section II. After that, a practical design of M$^2$I is discussed in Section III. Next, the wireless communication performance by using practical M$^2$I are presented in Section IV. Finally, this paper is concluded in Section V.

\section{Theoretical Modeling of M$^2$I }
\label{sec:review}

M$^2$I was proposed in our previous work \cite{guo2015m2i} to extend the magnetic induction communication range in complex environments, such as underwater and underground. The geometric structure is illustrated in Fig. \ref{fig:review1}. A loop antenna is enclosed by a metamaterial shell with thickness $r_2-r_1$. The permeability of the inner layer, metamaterial layer and outer layer are denoted by $\mu_1$, $\mu_2$ and $\mu_3$, respectively. Also, the metamaterial layer is considered as ideally homogeneous and isotropic. The key parameter of MI communication is the mutual inductance, which couples the transmitting and receiving antenna together. By using M$^2$I, the mutual inductance can be expressed as \cite{guo2015m2i} $M_0=\frac{\mathcal F_1}{\bar S_{m}^2}$, where
\begin{gather}
\label{equ:sm}
\textstyle
{\bar S_{m}}={\mathcal F_2}\underbrace{\left[2r_1^3(\mu_1-\mu_2)(\mu_3-\mu_2)-r_2^3(2\mu_2+\mu_1)(2\mu_3+\mu_2)\right]}_{\text{${\mathcal F_3}$}}+{\hat o},
\end{gather}
${\mathcal F_1}$ and ${\mathcal F_2}$ are coefficients, and ${\hat o}$ is an asymptotically small value. The first item on the right-hand side of \eqref{equ:sm} is much larger than ${\hat o}$. When ${\mathcal F_3}$ is zero, ${\bar S_m}$ can be minimized, which in turn maximizes $M_0$. Therefore, if the outer shell radius $r_2$, inner permeability $\mu_1$ and outer permeability $\mu_3$ are determined, by adjusting metamaterial shell's thickness $r_1$ and permeability $\mu_2$, we can always satisfy the condition to maximize $M_0$.

This enhancement is because of the matching between the metamaterial layer and the coil antenna. As depicted on the right-hand side in Fig. \ref{fig:review1}, a coil in conventional materials has a positive self-inductance $L$, while $L$ becomes negative when the coil is in metamaterials due to the negative permeability. Note that, the impedance of a coil in conventional materials is positive $j\omega L$, where $j=\sqrt{-1}$ and $\omega$ is the angular frequency. On the contrary, a coil in metamaterials has negative impedance, which is equivalent to a capacitor $C$ whose impedance is also negative $\frac{-j}{\omega C}$. Intuitively, we can use metamaterials to match with conventional materials to reduce the reactive power in vicinity of a loop antenna\cite{Ziolkowski_electric,erentok2008metamaterial}.

\section{Practical Design of M$^2$I Shell}
In this section, firstly, an approach is presented to achieve the negative permeability. Then, based on it we propose a spherical structure to obtain the negative permeability. Both the mathematical deductions and intuitive interpretations are provided.
\subsection{Negative Permeability}
\label{sec:negativemu}
When a mixture's components are much smaller than wavelength, the mixture can be regarded as an effectively homogeneous medium and its constitutive parameters, such as effective permeability, permittivity and conductivity, can be found \cite{EMT_theory}. First, by using a coil as an example, we show the way to obtain the effective permeability.

As shown in Fig. \ref{fig:ideal2}, a source radiates magnetic field $H_s {\hat s}$ which can induce current $I_s$ in a coil whose center is a point $o$. Here, ${\hat s}$ is a unit vector standing for the direction of magnetic field. If we consider there is no coil, at point $o$, the magnetic flux can be expressed as $B_m =\mu_0 H_s {\hat s}$, where $\mu_0$ is the permeability of vacuum. However, if the coil exits, due to the induced current $I_s$, the coil reradiates magnetic field $H_c {\hat c}$, where ${\hat c}$ is also a unit vector denoting the direction of the reradiated magnetic field. As a result, the magnetic flux at point $o$ can be updated as
\begin{gather}
\label{equ:effectivemu}
B_m=\mu_0 H_s {\hat s}+\mu_0 H_c {\hat c}=\mu_0\left(1+\frac{H_c {\hat c}}{H_s {\hat s}}\right)H_s {\hat s}=\mu_{eff}H_s {\hat s},
\end{gather}
where $\mu_{eff}$ is the effective permeability at point $o$.

From \eqref{equ:effectivemu} we can see, the effective permeability can be controlled provided that we can manipulate the reradiated magnetic field $H_c{\hat c}$. When $\frac{H_c {\hat c}}{H_s {\hat s}}$ is negative and its absolute value is larger than 1, the negative permeability can be obtained. The detailed approach to adjust $H_c{\hat c}$ is discussed in the following sections. Also, it's worth noting that besides coil, there are many other radiators can be utilized to achieve negative permeability, such as complicated metallic structures and dielectric spheres \cite{caloz2005electromagnetic}. Since coil is relatively easy to fabricate and its performance is tractable, we adopt it in this paper.
\subsection{Spherical M$^2$I Shell}

The geometrical structure of the spherical coil array is presented in Fig. \ref{fig:geo1}. The same as our discussion in \cite{guo2015m2i}, we set the outer radius of the spherical shell as 0.05~m. By using the Pentakis icosidodecahedron \cite{horn1984extended}, we equally divide the surface of the sphere into 80 triangles with 42 vertices. Each vertex is the center of a coil and the coil's orientation is the radial direction from the center of the shell to the vertex. In this way, the coil array is constructed by 42 identical coils. The realistic 3D model is demonstrated in Fig. \ref{fig:geo1}(b). The coils are made of copper wires with capacitors to tune them. In the equivalent circuit $L_c$, $R_c$, and $C$ represent the coil inductance, coil resistance, and capacitance, respectively. In the following, we proof that this structure can achieve negative permeability and it is equivalent to the ideal metamaterial shell in \cite{guo2015m2i}.
\subsubsection{Effective Permeability of M$^2$I Shell}
\begin{figure}[t]
  \centering
    \includegraphics[width=0.23\textwidth]{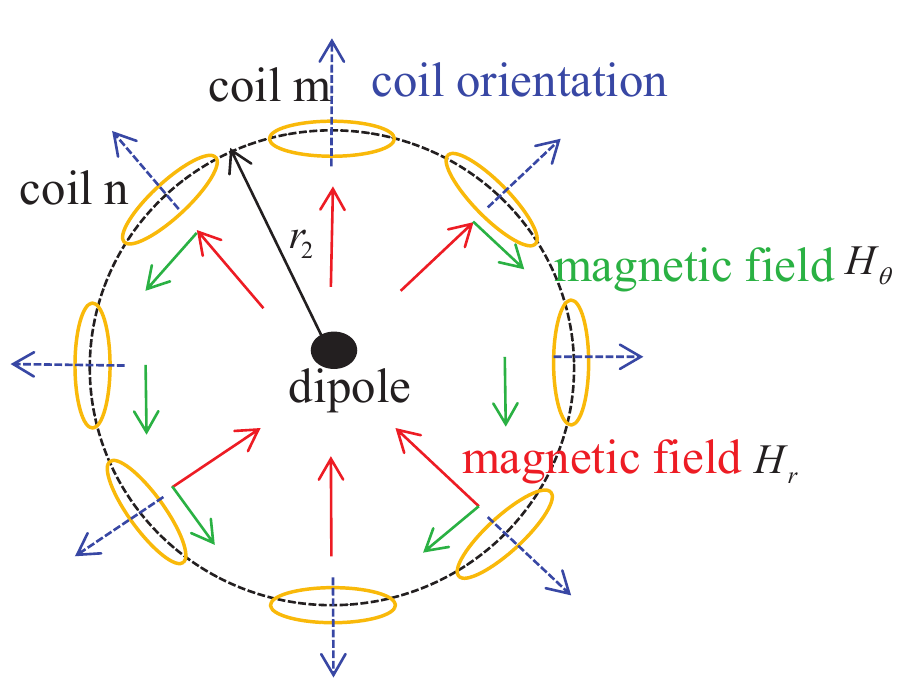}
    \vspace{-5pt}
  \caption{Illustration of coil orientation and the direction of magnetic field radiated by a magnetic loop antenna. }
  \vspace{-10pt}
  \label{fig:ideal12}
\end{figure}
First, the spherical coil array are considered as a homogeneous layer. By using Effective Medium Theory \cite{pendry1999magnetism}, we find the effective permeability of this layer. The radiation source is the same as our work in \cite{guo2015m2i}, i.e., a magnetic loop antenna.
All the coils on the shell are passive and they are excited by a loop antenna at the center of the shell, as shown in Fig. \ref{fig:ideal12}. According to Kirchhoff voltage law we can obtain,
\begin{equation}
\label{equ:app1}
\underbrace{\begin{bmatrix}
  Z_1 & j\omega M_{12} & \cdots&j\omega M_{1N} \\
  j\omega M_{21} &Z_2 & \cdots & j\omega M_{2N} \\
  \vdots & \vdots & \vdots & \vdots \\
  j\omega M_{N1} & j\omega M_{N2}& \cdots & Z_{N}
\end{bmatrix}}_{\text{${\bf Z}$}}
\cdot
\underbrace{
\begin{bmatrix}
  I_1 \\
  I_2 \\
  \vdots \\
  I_N
\end{bmatrix}}_{\text{${\bf I}$}}=
\underbrace{\begin{bmatrix}
  V_1 \\
  V_2\\
  \vdots \\
  V_N
\end{bmatrix}}_{\text{${\bf V}$}}
\end{equation}
where $N$ is the number of coils on the shell, $Z_n=R_c+j\omega L_c+1/j\omega C$, $M_{mn}$ is the mutual inductance between coil $m$ and coil $n$ and the detailed calculation is provided in Appendix A, $I_n$ is the induced current in coil $n$, and $V_n$ is the voltage. The magnetic field radiated by a coil or loop antenna can be expressed as \cite{Balanis_a}
\begin{align}
\label{equ:antenna_field}
H_r(I_0)&={\hat H_r}\cos\theta{\hat r_0}=\frac{j k a^2 I_0 \cos{\theta}}{2d^2}\left[1+\frac{1}{jkd}\right]e^{-jkd}{\hat r_0},\\
H_t(I_0)&=\frac{- k^2 a^2 I_0 \sin{\theta}}{4d}\left[1+\frac{1}{jkd}-\frac{1}{(kd)^2}\right]e^{-jkd}{\hat \theta_0},
\end{align}
where $a$ is the radius of a coil, $k$ is the wavenumber, $d$ is the distance from the coil's center, $\theta$ is the azimuthal angle, and ${\hat r_0}$ and ${\hat \theta_0}$ are unit vectors representing the radial and azimuthal direction, respectively. The direction of $H_r$ and $H_t$ are depicted in Fig. \ref{fig:ideal12}. Since the coil orientation is parallel with $H_r$'s direction and perpendicular to $H_t$'s direction, only $H_r$ can induce currents in the coils on the shell. As a result, this shell can be regarded as metamaterial for $H_r$, but not for $H_t$. In our future work, a tri-directional coil \cite{Guo_oil_2014} will be adopted to make the coils on the shell isotropic but it's out of the scope of this paper. Accordingly, $V_n$ can be expressed as
\begin{align}
\label{equ:voltage}
V_n=-j\omega\pi a^2 \mu_0 {\hat H_r} \cos{\theta_n},
\end{align}
where ${\hat H_r}$ can be derived from \eqref{equ:antenna_field} and it is the same for all the coils on the shell. According to Kirchhoff's voltage law, without loss of generality, in coil $n$ we can obtain
\begin{gather}
\label{equ:kirchhoff}
\textstyle
I_n(R_c+j\omega L_c-j/\omega C)+\sum_{i=1,i\neq n}^{N}j\omega M_{in}I_i=-j\omega\pi a^2 \mu_0 {\hat H_r} \cos \theta_n.
\end{gather}
By rearranging \eqref{equ:kirchhoff} we can obtain,
\begin{gather}
\label{equ:rearrange1}
\frac{I_n(R_c+j\omega L_c+1/j\omega C)}{{\hat H_r \cos \theta_n}}=-j\omega\pi a^2 \left(\sum_{i=1,i\neq n}^{N}\frac{M_{in}I_i}{\pi a^2 {\hat H_r \cos\theta_n}}+\mu_0\right).
\end{gather}
In addition, we have the following proposition:
\begin{proposition}
\label{pro:IH}
For any coil $n$ on the shell, $I_n={\mathcal A} {\hat H_r} \cos \theta_n$, where ${\mathcal A}$ is a constant.
\end{proposition}
\begin{proof}
See Appendix B.
\end{proof}

If only coil $n$ exists on the spherical shell and all other coils are removed, there is only $\mu_0$ in the bracket on the right-hand side of \eqref{equ:rearrange1}. In other words, the first item in the bracket denotes the mutual interactions among coils on the shell. When there are no other coils, this item does not exist and this is indeed the conventional material with permeability $\mu_0$. Hence, when we consider all other coils on the shell, the terms in the bracket can be regarded as the effective permeability. In addition, based on Proposition~\ref{pro:IH}, since $\frac{I_n}{{\hat H_r} \cos \theta_n}={\mathcal A}$, the left-hand side of \eqref{equ:rearrange1} is a constant. Therefore, no matter which coil on the shell we choose, the value in the bracket does not change, i.e., the effective permeability is homogeneous.

By using \eqref{equ:rearrange1} and the value of ${\mathcal A}$ in Appendix B, we can obtain
\begin{gather}
\label{equ:effmu}
\textstyle
\mu_{eff}=\left(1-\frac{j\omega \sum_{i=1,i\neq n}^{N}M_{in}}{R_c+j\omega L_c+1/j\omega C+\sum_{i=1,i\neq n}^{N}j\omega M_{in}}\right)\mu_0.
\end{gather}
From \eqref{equ:effmu} we can see when the second item in the bracket is larger than 1, a negative permeability can be achieved. Moreover, from \eqref{equ:effmu} we can find that the effective permeability is only determined by the coils' mutual inductances and the lumped elements. It is not affected by the source.

Based on \eqref{equ:effmu}, we numerically evaluate the effective permeability of the spherical coil array. The numerical parameters of the system are provided in Table I. Note that in Table I, $L_{app}=\frac{\pi \mu_0 a}{2}$ is the approximated self-inductance and $L_c$ is the accurate value measured in COMSOL. $L_{app}$ is utilized to obtain a succinct closed-form formula to elucidate physics better. In the following, we will show that this $L_{app}$ only changes the resonant frequency.
\begin{table}
\centering
\renewcommand{\arraystretch}{1.3}
\caption{Simulation Parameters}
\label{table_parameters}
\begin{tabular}{c c|c c|c c}
\hline
$\mu_0$ & $4\pi\times10^{-7}~$H/m& $\epsilon_0$&$8.854\times10^{-12}$~F/m&$L_{app}$&13.8~nH\\
\hline
$f$& 10~MHz& $a$&0.007~m&$L_{c}$&15.2~nH\\
\hline
$r_2$&0.05~m & $R_c$&0.244~m$\Omega$\\
\hline
\end{tabular}
     \vspace{-8pt}
\end{table}

\begin{figure}[t]
  \centering
  \subfigure[Effect of capacitance]{
    \label{fig:mueff_cap}
    \includegraphics[width=0.21\textwidth]{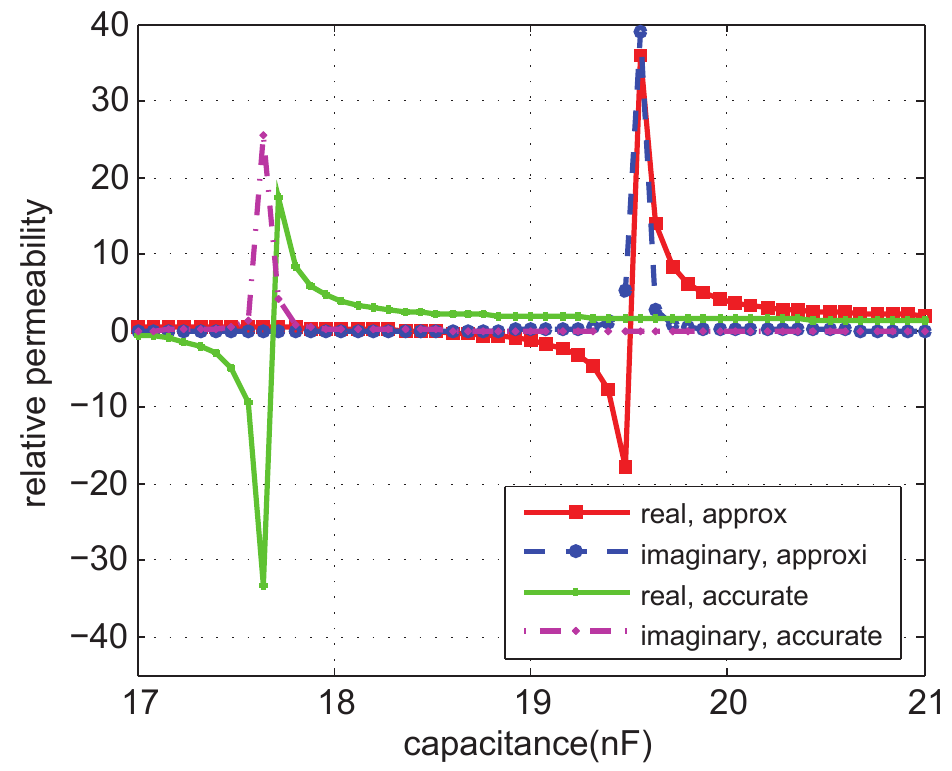}}\quad
  \subfigure[Effect of frequency]{%
    \includegraphics[width=0.21\textwidth]{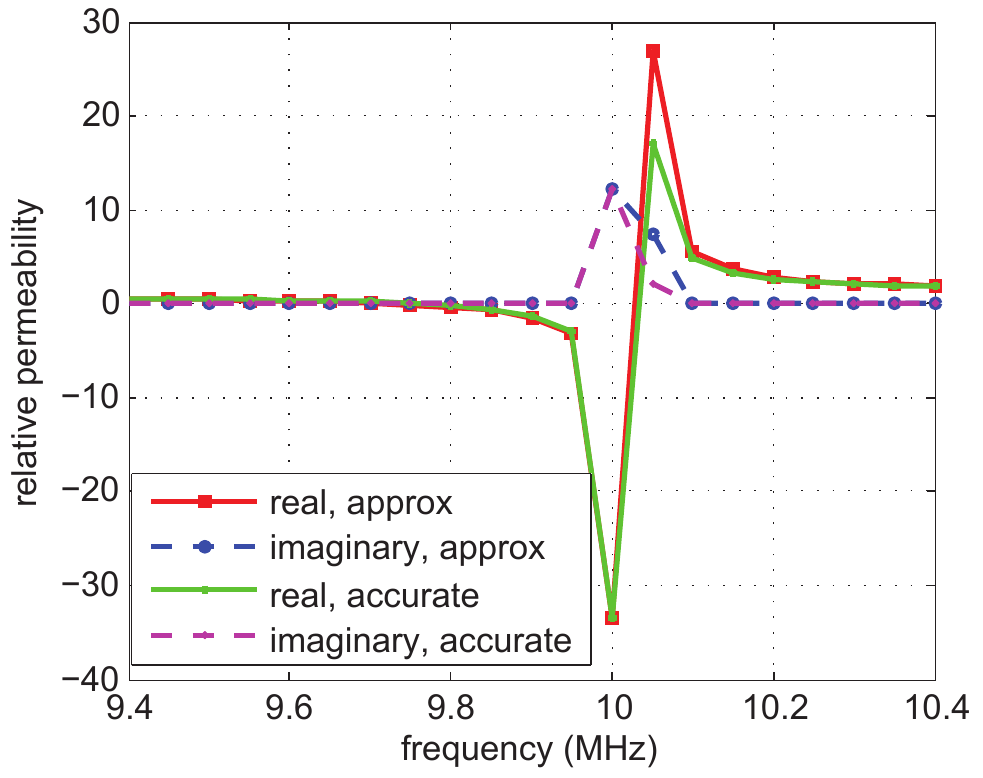}\quad
    \label{fig:mueff_f}}
      \vspace{-5pt}
  \caption{Effective permeability. }
    \vspace{-10pt}
  \label{fig:cf}
\end{figure}

The effective permeability is shown in Fig. \ref{fig:cf}. Due to the resistance of coils, the effective permeability is a complex number where the imaginary part stands for the loss in coils. First, we set the frequency as 10~MHz to investigate the effect of capacitance. As depicted in Fig. \ref{fig:mueff_cap}, the approximated and accurate self-inductance has the same trend but there is a little shift. Moreover, when the capacitance is small, the impedance of the capacitor is large since $Z_c=\frac{1}{j\omega C}$. Therefore, the relative permeability is negative as discussed in Section \ref{sec:review}. When the capacitance is large, the inductance is dominant. Then, the relative permeability becomes positive. Also, there is a dramatic change from negative to positive. The reason is that around the resonant point, the capacitor and inductor are almost perfectly matched. Hence, the reactance is very small and the current in each coil is significant. A slightly change of the capacitance or inductance can greatly change the current direction in the coil since an inductor adds $\pi/2$ to the current's phase while a capacitor reduces $\pi/2$ of the phase. As a result, the relative permeability changes its sign.

Then, we set the capacitance as 17.6~nF for the accurate self-inductance and 19.5~nF for the approximated self-inductance to investigate the effect of frequency. As shown in Fig. \ref{fig:mueff_f}, by using the approximated and accurate self-inductance, they have almost the same performance since both of them are tuned at 10~MHz. Furthermore, a large negative effective permeability is obtained at 10~MHz. However, the imaginary part of the permeability is also very large which reduces the merit of this point. Also, we notice that changing capacitance is similar as changing frequency since in \eqref{equ:effmu}, $\omega$ and $C$ have the same effect on the impedance of the capacitor.

\subsubsection{Variation of Induced Current Direction}
\label{sec:reradiated}

\begin{figure}[t]
  \centering
    \includegraphics[width=0.18\textwidth]{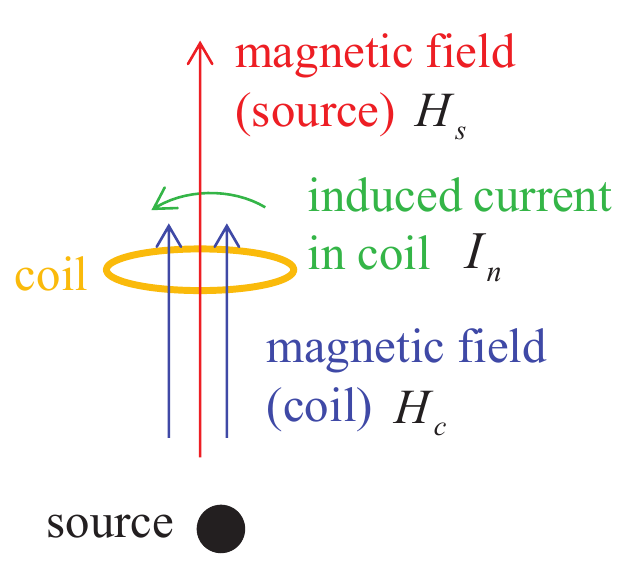}
    \vspace{-5pt}
  \caption{Induced current in a coil and reradiated magnetic field. }
  \vspace{-10pt}
  \label{fig:ideal3}
\end{figure}

Based on our observations, we provide more insightful discussions to understand the negative permeability intuitively. As we have discussed in Section \ref{sec:negativemu}, by manipulating the reradiated magnetic field with coils we can change the effective permeability. Next, we show how the direction of the reradiated magnetic field is changed by adjusting the capacitor.

Consider that the equivalent circuit of a coil consists of coil resistance ${\hat R_c}$, coil inductance ${\hat L_c}$, and a tunable capacitor ${\hat C}$. The normal-direction incoming magnetic field is $H_s$ and the induced voltage on the coil is $V=-j\omega \pi a^2 \mu_0 H_s=jV_0$. Now we consider three scenarios to distinguish our approach from existing works and they are summarized as follows:
\begin{align}
\label{equ:largeL}I_c&=\frac{V}{{\hat R_c}+j\omega {\hat L_c}}\approx \frac{V_0}{\omega {\hat L_c}},when~\omega {\hat L_c}>\frac{1}{\omega {\hat C}};\\
\label{equ:equLC}I_c&=\frac{V}{{\hat R_c}}=j\frac{V_0}{{\hat R_c}},~when~\omega {\hat L_c}=\frac{1}{\omega {\hat C}};\\
\label{equ:largeC}I_c&=\frac{V}{{\hat R_c}+\frac{1}{j \omega {\hat C}}}\approx -V_0\omega {\hat C},~when~\omega {\hat L_c}<\frac{1}{\omega {\hat C}}.
\end{align}
As shown in \eqref{equ:equLC}, the current is shifted 90 degree from \eqref{equ:largeL} and in \eqref{equ:largeC} the current is further shifted 90 degree which is the opposite direction of the first scenario. Since the magnetic field reradiated by the first scenario is counter direction as $H_s$, the third scenario would reradiate the same direction magnetic field as $H_s$, as shown in Fig. \ref{fig:ideal3}. Therefore, by changing the capacitor, we can change the induced magnetic field direction to obtain the desired effective permeability. We consider the coil has low resistance, i.e., $R_c$ is much smaller than $\omega L$ and $1/\omega C$. It should be noted that this assumption is necessary and practical since the power dissipated in metamaterial is due to this resistance. Referring to our analysis in \cite{guo2015m2i}, the high loss metamaterial is hard to enhance the radiated magnetic field. In the full-wave simulation, the coil is made of low resistance copper which has low cost.

\section{Characterizing M$^2$I Communication under Practical Design}
In this section, the communication performance of spherical coil array-based M$^2$I is evaluated and compared with the original MI. The key factor of MI communication, i.e., magnetic field intensity, is discussed first. Then, the path loss, bandwidth and channel capacity of M$^2$I channel are presented.

Note that, the radius of the loop antenna inside a spherical coil array is set as 0.025~m, while the radius of MI loop antenna is set as 0.05~m, i.e., the outer radius of the coil array shell. In other words, they occupy the same space. Also, the ideal M$^2$I's performance is not comparable with coil array-based M$^2$I even they have the same negative permeability. Because the inner radius of ideal M$^2$I in \cite{guo2015m2i} is 0.025~m and the loop antenna is 0.015~m in radius. However, the coil array-based M$^2$I has almost no physical thickness. As a result, we can set the loop antenna much larger than ideal M$^2$I to fully use the space. Even the physical thickness of practical M$^2$I is small, in this section, we would show that the effective thickness and permeability are equivalent to the lumped capacitor, which can be adjusted to find the optimal configuration.
\subsection{Magnetic Field Analysis}
\subsubsection{Radiated Magnetic Field}
In this part, we put a loop antenna into the shell to evaluate its performance. The geometric model of both theory and simulation is demonstrated in Fig. \ref{fig:comgeo}. Only the large loop antenna is actively excited and all other small coils on the shell are passive.
\begin{figure}[t]
  \centering
  \subfigure[Simulation space]{
    \label{fig:comgeo1}
    \includegraphics[width=0.215\textwidth]{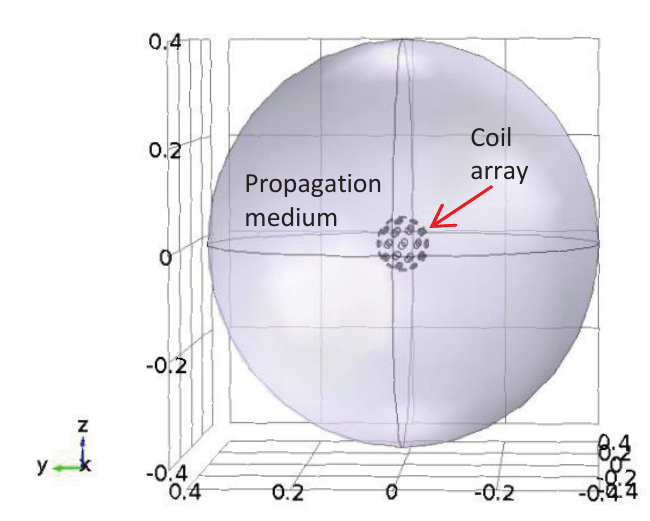}}\quad
  \subfigure[M$^2$I antenna]{%
    \includegraphics[width=0.215\textwidth]{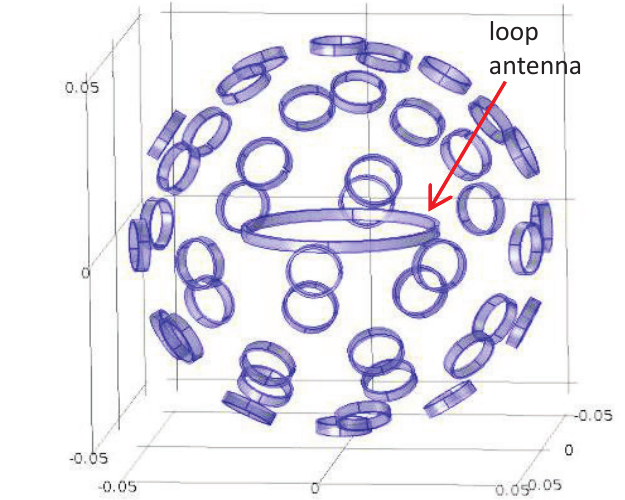}\quad
    \label{fig:comgeo2}}
      \vspace{-5pt}
  \caption{Geometry of full-wave simulation in COMSOL Multiphysics. }
    \vspace{-10pt}
  \label{fig:comgeo}
\end{figure}

\begin{figure}[t]
  \centering
    \includegraphics[width=0.4\textwidth]{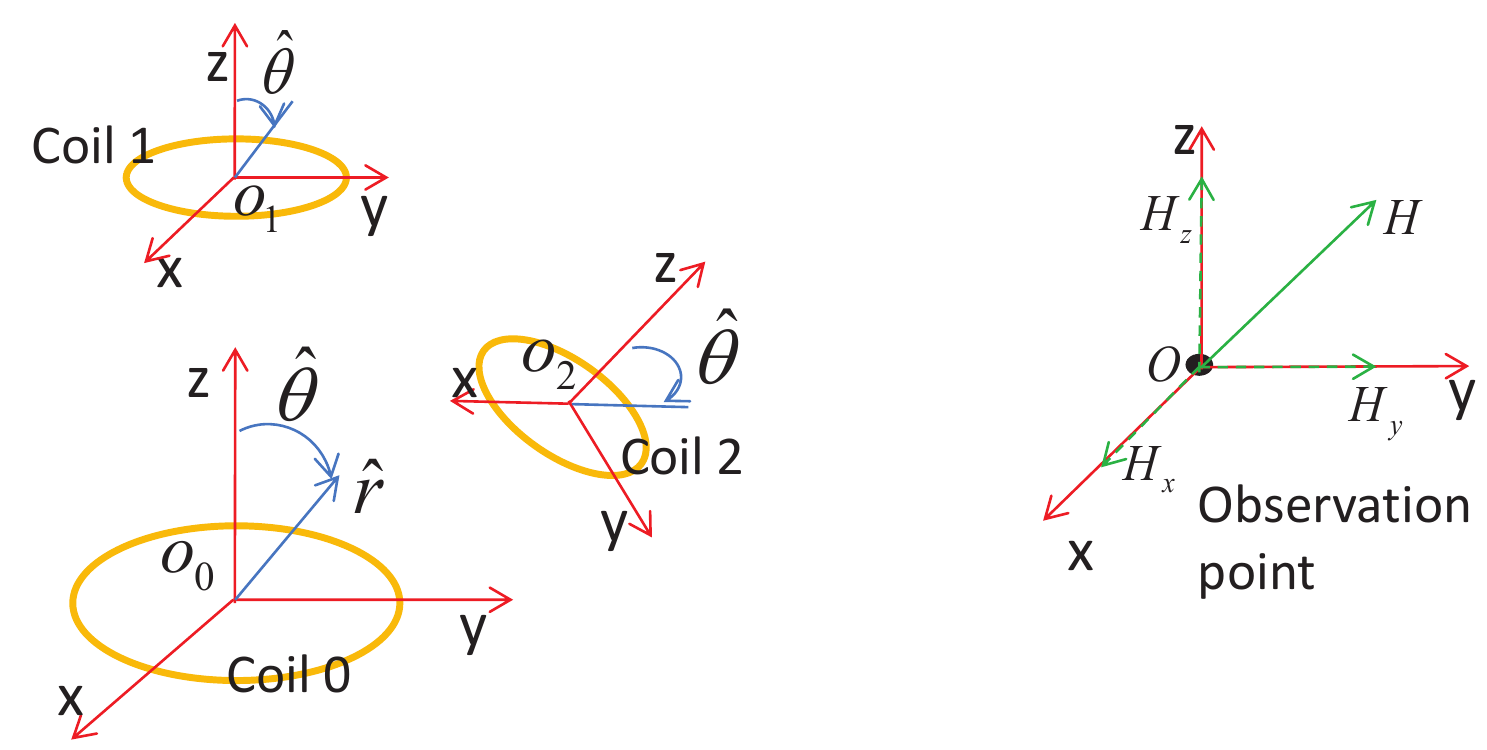}
    \vspace{-5pt}
  \caption{Coils' coordinates and field direction at an observation point. }
  \vspace{-10pt}
  \label{fig:direction}
\end{figure}

Since the induced power in the receiving coil is proportional to magnetic field intensity, we use this intensity as a metric to show the enhancement. Note that \eqref{equ:antenna_field} can be applied when the coordinates origin is the center of the coil and z-axis overlaps the coil's orientation vector. For instance, in Fig. \ref{fig:direction}, the three coils have their own coordinates with different orientations. At the observation point $O$, we first consider all the fields from each antenna individually. Then we decompose the fields in a Cartesian coordinates. Finally, we add all the fields together. As a result, the magnetic field at point $O$ can be found.

The radiated magnetic field is the summation of the radiated magnetic field by the loop antenna and reradiated magnetic fields by the passive coils on the shell, which can be expressed as
\begin{gather}
\label{equ:field}
\textstyle
H=\sum_{i=0}^{N}\left[H_r(I_i){\hat r_i}+H_t(I_i){\hat \theta_i}\right],
\end{gather}
where the large loop antenna is denoted by $i=0$. Then, the induced current can be determined by using \eqref{equ:app1}.

\subsubsection{Magnetic Field Enhancement}
\begin{figure}[t]
  \centering
  \subfigure[0.2~m]{
    \label{fig:fc02}
    \includegraphics[width=0.215\textwidth]{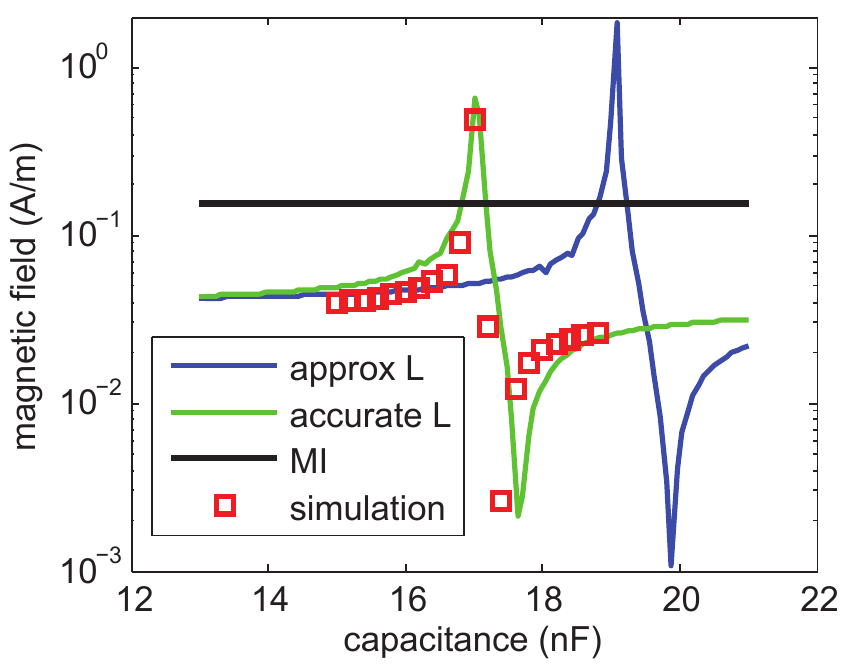}}\quad
  \subfigure[0.4~m]{%
    \includegraphics[width=0.215\textwidth]{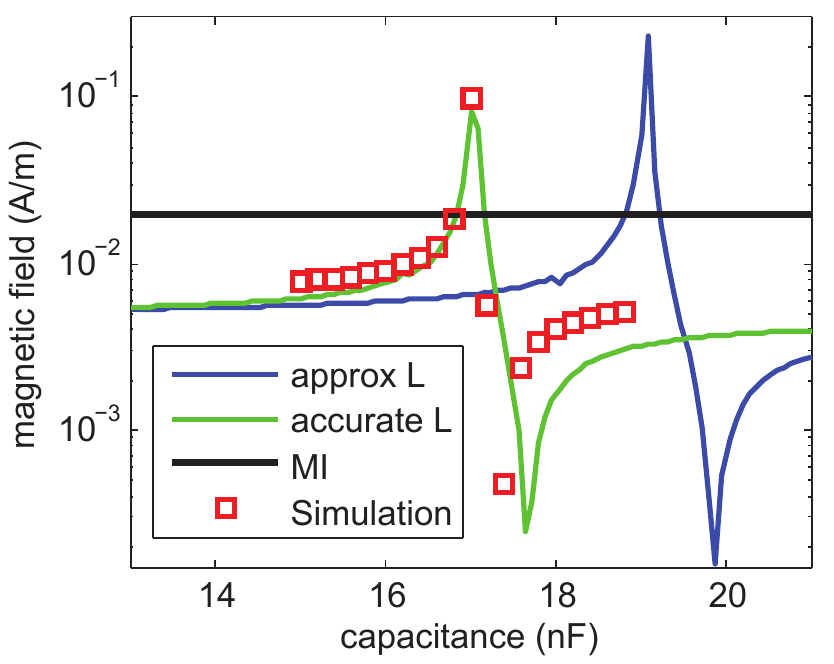}\quad
    \label{fig:fc04}}
      \vspace{-5pt}
  \caption{Magnetic field (A/m) at a distance 0.2~m and 0.4~m from the center of the shell. }
    \vspace{-10pt}
  \label{fig:fc0204}
\end{figure}

The radiated field is evaluated in this part. The radius of the loop antenna is set as 0.025~m and its current is set as 1~A. The MI antenna with a larger radius 0.05~m, i.e., the outer radius of the shell, is excited with the same current. Then, we compare the radiated magnetic field intensity at 0.2~m and 0.4~m from the center of M$^2$I and MI loop antennas. Since the dipole moment is $\pi a^2 I_0$, the MI loop antenna's dipole moment is 4 times larger than M$^2$I loop antenna. The effect of $C$ on the field intensity is investigated. Also, the magnetic field is measured along z-axis since $H_r$ is dominant.

As shown in Fig. \ref{fig:fc0204}, by using the accurate self-inductance, the theoretical result agrees well with the full-wave simulation. Also, by using the approximated self-inductance, it does not affect the trend of the resonance. The coil array-based M$^2$I can provide almost one order gain. Referring to Fig. \ref{fig:mueff_cap}, it is interesting to find that the highest magnetic field intensity does not appear at the most negative or positive effective permeability. The strongest field intensity is obtained when the effective permeability is a little smaller than 0. The reason can be interpreted from two aspects. On the one hand, we need a negative permeability to match with the inductive source. According to \eqref{equ:resonance}, by adjusting $C$, we can always find the resonance. On the other hand, when the effective permeability becomes more negative, the imaginary part of the permeability increases dramatically which deteriorates the performance. Therefore, we have to strike a balance between the negative permeability and low loss, i.e., small imaginary part of the effective permeability. Additionally, by comparing Fig. \ref{fig:fc02} and Fig. \ref{fig:fc04}, we find that the gain from this coil array-based M$^2$I does not change with distance which can validate that this gain is from radiation, but not near field induction. In Fig. \ref{fig:dm_distance}, we compare the radiated field by M$^2$I coil array (outer radius is 0.05~m) and original large MI antenna (radius is 0.05~m). Due to the high computation burden in COMSOL, we show the simulated magnetic field intensity within 1~m. The results further proofs that the gain is a constant with distance change.


\subsubsection{Resonance Condition}
Moreover, the resonant capacitor value can be estimated. In order to enhance the magnetic field, we need to maximize $I_n$ to enlarge the reradiated field by the coils. In \eqref{equ:effmu}, the second item in the bracket can be maximized when the reactive terms on the denominator are eliminated. Therefore, by letting $j\omega L_c+1/j\omega C+\sum_{i=1,i\neq n}^{N}j\omega M_{in}=0$ and $L_c\approx L_{app}=\frac{\mu \pi a}{2}$, we can obtain
\begin{gather}
\label{equ:resonance}
\frac{a}{r_2}\approx\sqrt[3]{\frac{1}{15.3}\left(\frac{1}{2}-\frac{1}{\omega^2 \mu_0 \pi a C}\right)}.
\end{gather}

Referring to equation (13) in \cite{guo2015m2i}, by adjusting the thickness and negative permeability, we can achieve a resonance to amplify the magnetic field intensity. In this paper, when all other configurations are fixed, by adjusting the capacitor $C$ we can also achieve this resonance, as suggested by \eqref{equ:resonance}. Since the negative permeability metamaterial is equivalent to capacitance as discussed in Section \ref{sec:review}, they have the same physical principles. Therefore, changing $C$ is equivalent to changing both the thickness and the negative permeability.

\subsubsection{Negative Self-Inductance}

\begin{figure*}[t]
\begin{minipage}{0.23\textwidth}
\centering
  \includegraphics[width=0.9\textwidth]{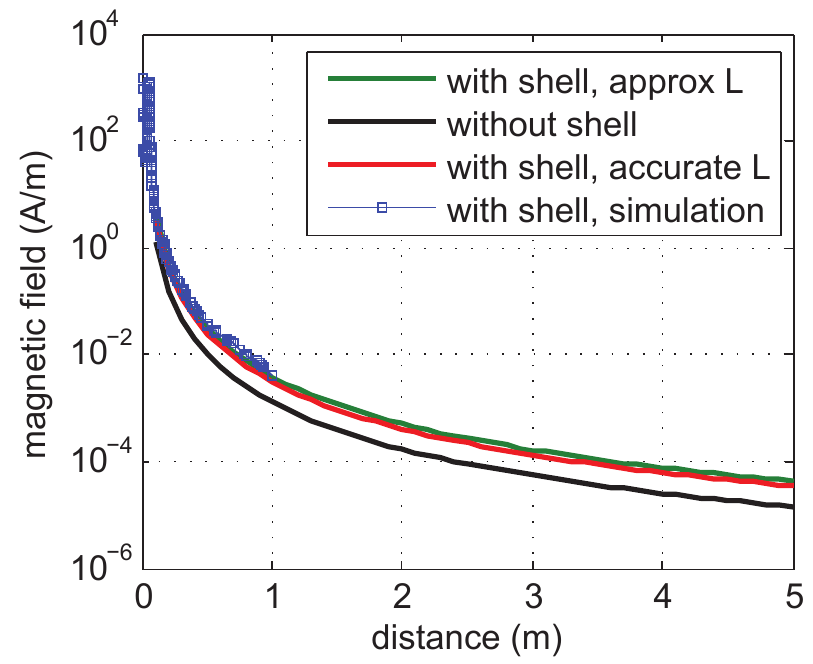}
  \caption{Distance effect on Magnetic field intensity.}
    \label{fig:dm_distance}
\end{minipage}\quad
\begin{minipage}{0.7\textwidth}
\centering
  \subfigure[17~nF]{
    \label{fig:fd0}
    \includegraphics[width=0.3\textwidth]{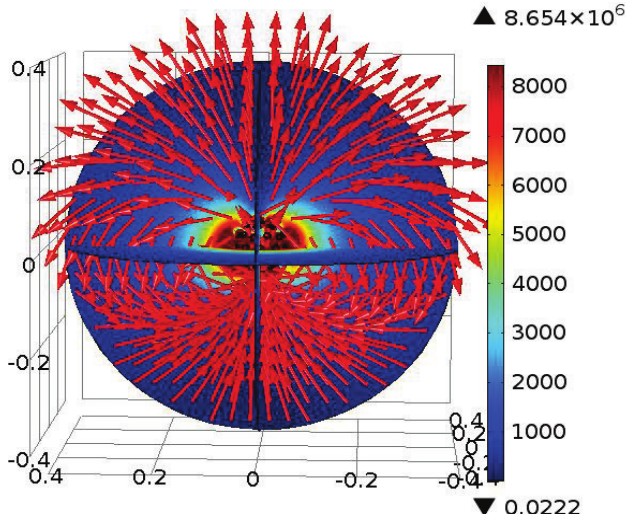}}\quad
  \subfigure[17.2~nF]{%
    \includegraphics[width=0.3\textwidth]{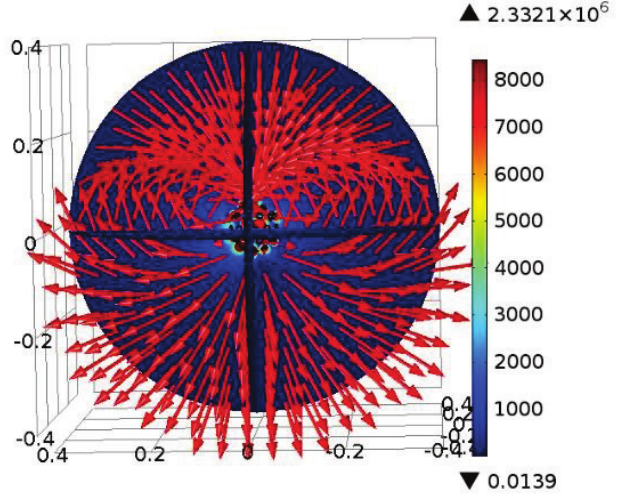}\quad
    \label{fig:fd1}}
    \subfigure[17.5~nF]{%
    \includegraphics[width=0.3\textwidth]{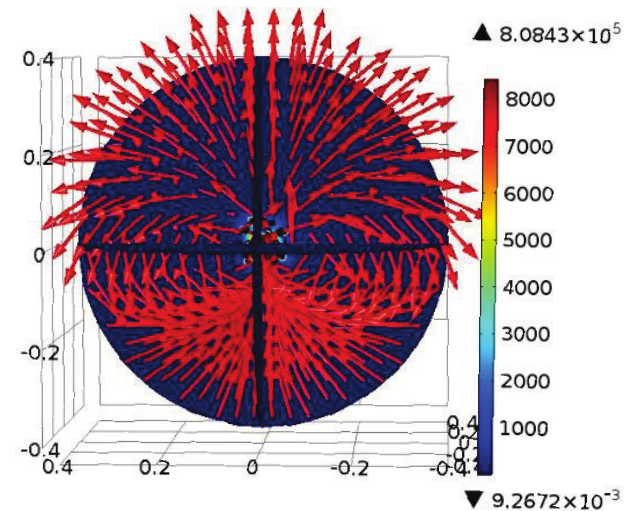}\quad
    \label{fig:fd2}}
  \caption{Simulation results of magnetic field direction close to resonance condition.}
    \vspace{-10pt}
  \label{fig:fd}
\end{minipage}\vspace{-10pt}
\end{figure*}
It is also interesting to notice that there is a dramatic decrease of field intensity after the resonance point in Fig. \ref{fig:fc0204}, i.e., around 17.5~nF for the accurate $L$ and 20~nF for the approximated $L$. In \cite{guo2015m2i}, we found a negative self-inductance and the reason is explained theoretically. The results are also validated in Fig. 14 and Fig. 16 in \cite{guo2015m2i}. Since the self-inductance changes from positive to negative then it changes back from negative to positive,
there exists two points where the self-inductance is almost 0. Here, by using a practical M$^2$I, we find the same negative self-inductance and the dramatic decrease of magnetic field intensity is because of it. In the following, we explain the reason of this negative self-inductance from the perspective of practical design instead of ideal metamaterial.

By using the full-wave simulation results in COMSOL, which are shown in Fig. \ref{fig:fd}, we explain how the negative self-inductance is generated. When $C$=17~nF, the magnetic field intensity is large and the self-inductance is positive since the reradiated magnetic field and the magnetic field radiated by the loop antenna have the same direction.  However, if we increase $C$ to 17.2~nF, the magnetic field changes its direction and the self-inductance becomes negative. Referring to Fig. \ref{fig:ideal2}, when $H_c$ is larger than $H_s$, the overall magnetic field changes its direction from $H_s$ to $H_c$. In other words, when $1/j\omega C$ is smaller than $j\omega L$ and the the reradiated magnetic field from the shell is larger than the original magnetic field radiated by the loop antenna, the total field radiated from the coil array will change its direction. Note that, here Right-hand rule no longer validates, since a positive current in the loop antenna generates a magnetic field obeys Left-hand rule. Therefore, the self-inductance is negative. When we further increase $C$ to 17.5~nF, the self-inductance changes to positive again since the resonance does not exist and the currents in the coils are much smaller than the current in the loop antenna. Thus, when the self-inductance gradually changes from negative to positive, there is a region where its value is around zero which makes the radiated field extremely small. However, the change from positive to negative is dramatic and this effect is not obvious. The effective self-inductance $L_{eff}=L_{app}+\Im({Z_{ref}})/j\omega$, where $\Im$ denotes the imaginary part of a complex number, is also plotted in Fig. \ref{fig:inductance}. It shows the negative self-inductance is at around 19~nF. Due to the coupling among the loop antenna in the shell and coils on the shell, the reflected impedance in the loop antenna need to be considered which can be expressed as~\cite{Sun_MI_TAP_2010} $Z_{ref}=\sum_{i=1}^{N} \frac{\omega^2 M_{0i}^2}{Z_i}$, where the subscript 0 denotes the loop antenna. The negative self-inductance generated by the spherical coil array is consistent with the discussion in \cite{guo2015m2i}, i.e., Fig. 14 and Fig. 16, by using the ideal metamaterial.

Moreover, since we measure the magnetic field along z-axis, readers may wonder the two coils on z-axis may take dominant effect and other coils can be neglected. In Fig. \ref{fig:two_coil}, we only keep the two coils on z-axis and remove other coils. The capacitance $C$ is set as 17~nF. From the figure we can see the magnetic field intensity generated by the loop antenna is much smaller than that in Fig. \ref{fig:fd0} which can prove the contribution from the spherical coil array.
\begin{figure}[t]
\begin{minipage}{0.23\textwidth}
\centering
  \includegraphics[width=0.9\textwidth]{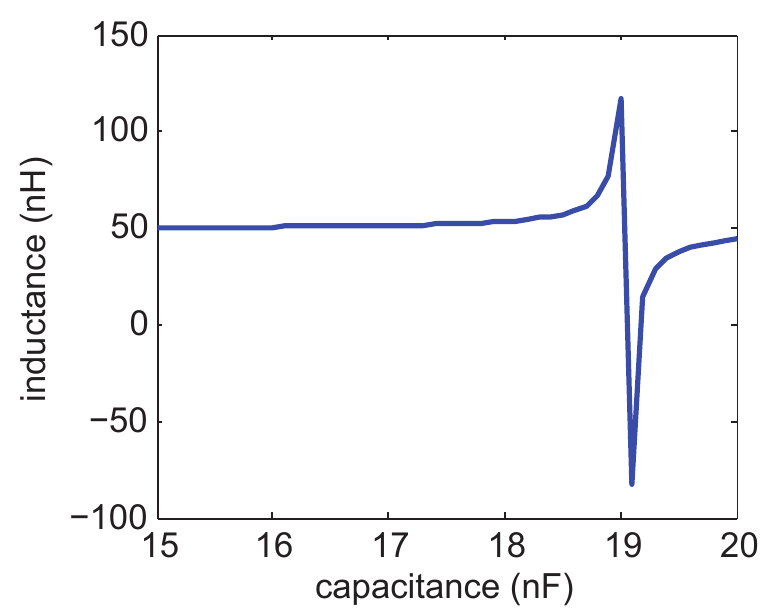}
  \caption{Self-inductance}
    \label{fig:inductance}
\end{minipage}\quad
\begin{minipage}{0.23\textwidth}
\centering
    \includegraphics[width=0.9\textwidth]{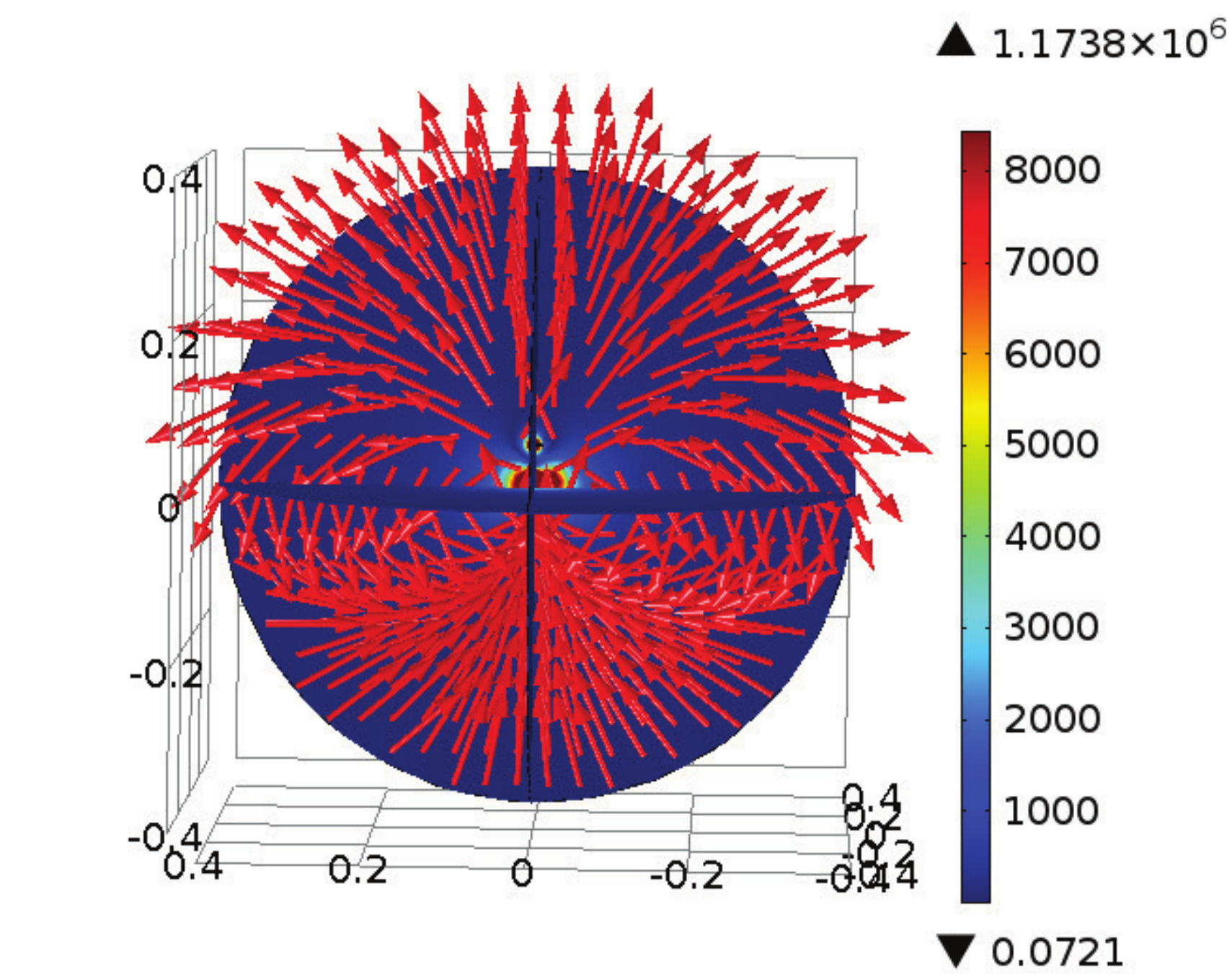}\quad
    \caption{Two coil on the shell}
    \label{fig:two_coil}
\end{minipage}
    \vspace{-10pt}
\end{figure}

\subsection{Wireless Channel Analysis}
Up to this point, we have demonstrated that the coil array-based M$^2$I can significantly enhance the radiated magnetic field intensity. In this section, we present the performance of the wireless channel between two M$^2$I transceivers, i.e., both the transmitter and receiver are equipped with coil array-based M$^2$I antenna.. In particular, the path loss, bandwidth and channel capacity are discussed.
\begin{figure*}[t]
  \centering
  \subfigure[Equivalent circuit]{
    \label{fig:circuit}
    \includegraphics[width=0.26\textwidth]{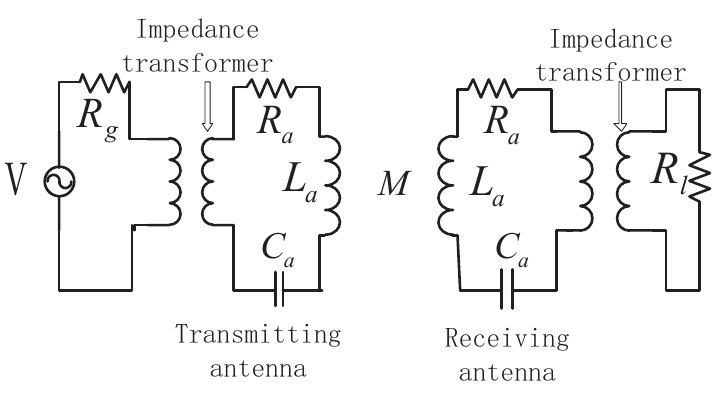}}\quad
  \subfigure[Path loss]{
    \label{fig:pathloss}
    \includegraphics[width=0.215\textwidth]{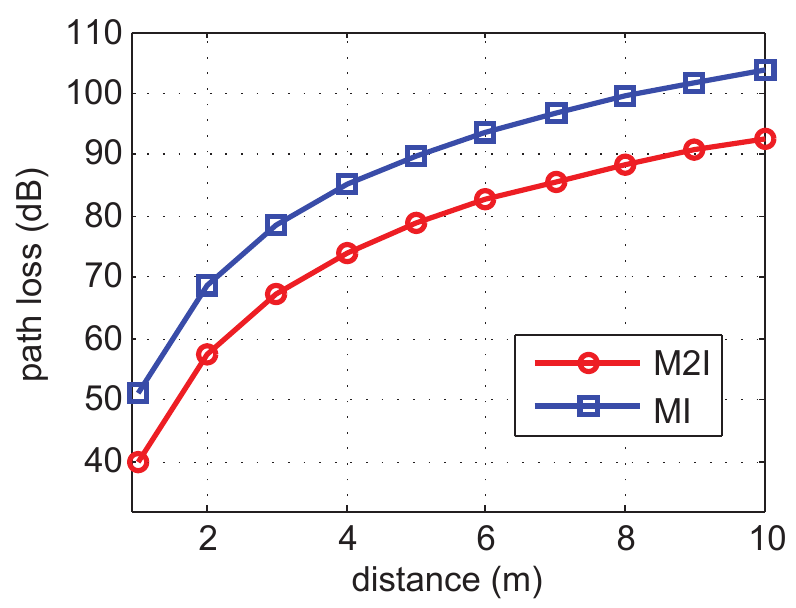}}\quad
  \subfigure[Frequency response]{%
    \includegraphics[width=0.2\textwidth]{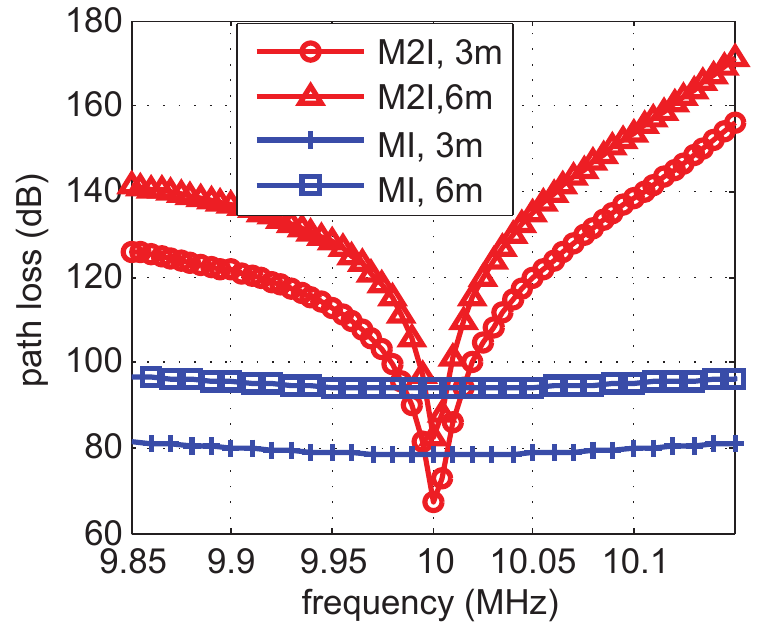}\quad
    \label{fig:bw}}
    \subfigure[Channel capacity]{%
    \includegraphics[width=0.2\textwidth]{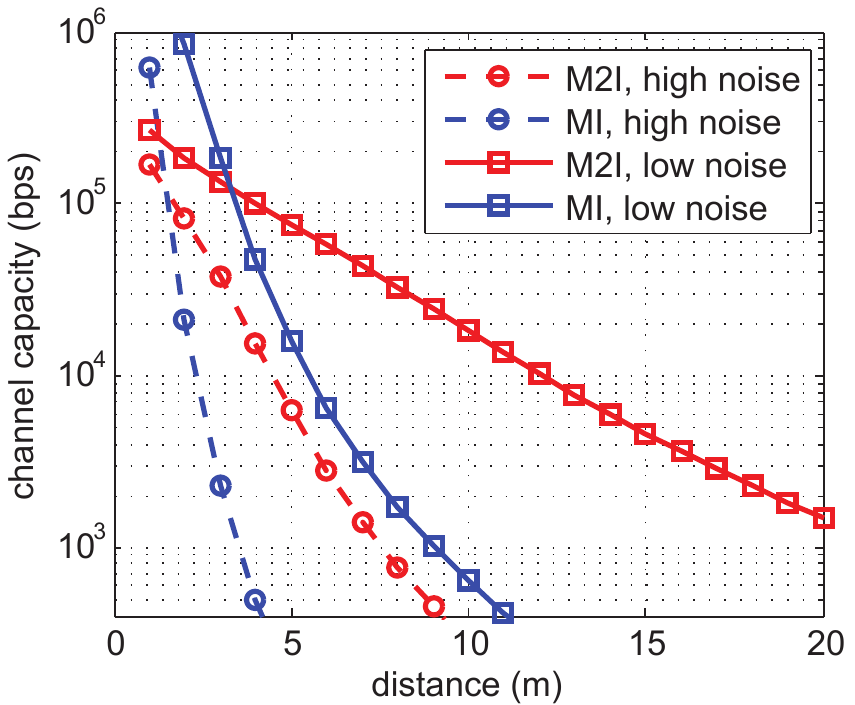}\quad
    \label{fig:capacity}}
      \vspace{-5pt}
  \caption{Communication performance. }
    \vspace{-10pt}
  \label{fig:communication}
\end{figure*}
\subsubsection{Path Loss}
The equivalent circuit is shown in Fig. \ref{fig:circuit}, where $R_g$ and $R_l$ are the source's output resistance and load resistance, respectively. Both of the $R_g$ and $R_l$ are set as a typical value 50~$\Omega$. The impedance transformer is utilized for impedance matching between the source/load and antenna. The resistance, self-inductance and tunable capacitance in the antenna circuit are denoted by $R_a$, $L_a$, and $C_a$, respectively.  When we compare the M$^2$I transceivers and original MI transceivers, the same as previous discussions, the radius of the original MI antenna is the outer radius of the shell, i.e., 0.05~m. In order to calculate the magnetic field at the receiving side, we need to add the passive coils in the receiving M$^2$I antenna into \eqref{equ:app1} and \eqref{equ:field}. Now, we have a transmitting antenna, 84 passive coils, and a receiving antenna. Therefore, the updated dimensions of ${\bf Z}$, ${\bf I}$ and ${\bf V}$ in \eqref{equ:app1} are 86$\times$86, 86$\times$1, and 86$\times$1, respectively.

Then the inductance in the loop antenna is matched by a tunable capacitor $C_a$. The resistance is matched by the impedance transformer:
\begin{gather}
{N_1^2}\cdot{[R_a+\Re(Z_{ref})]}={R_g}\cdot{N_2^2},
\end{gather}
where $\Re$ stands for the real part of a complex number, and $N_1$ and $N_2$ are the primary and secondary turns, respectively. Therefore, the voltage for the transmitting loop antenna is $V_0=V N_2/2N_1$. All other voltages in the updated \eqref{equ:app1} are 0. By solving the updated \eqref{equ:app1}, the induced current in antennas can be found. The received power can be expressed as $P_r=|I_l|^2R_l$, where $I_l$ is the load current. Similarly, the dissipated power in the source can be expressed as $P_t=V^2/2R_g$. Then, we can obtain the path loss ${\mathcal L}(d)=-10lg(P_r/ P_t)$.

Next, by using the approximated self-inductance $L_{app}$ and letting $C$=19.06~nF (optimal value find in Fig. \ref{fig:fc0204}), we compare the performance of coil array-based M$^2$I with that of the original MI. Also, the loop antenna resistance are set as 0.1~$\Omega$ and 0.2~$\Omega$ for coil array-based M$^2$I (antenna radius 0.025~m) and the original MI (antenna radius 0.05~m), respectively. As shown in Fig. \ref{fig:pathloss}, the path loss of M$^2$I is much lower than that of original MI which is consistent with~\cite{guo2015m2i}.
\subsubsection{Bandwidth}
In order to investigate the frequency response of the system, we keep the distance between a transmitter and a receiver as 3~m and 6~m and vary the frequency. As shown in Fig. \ref{fig:bw}, the enhancement of M$^2$I greatly reduces the bandwidth since magnetic field intensity is strongly based on resonance. The 3 dB bandwidth for M$^2$I is 4~kHz, while it is 320~kHz for original MI. Also, we notice that changing frequency is consistent with changing capacitance. When frequency is higher than the designed frequency, the signal strength drops fast. Moreover, at 3~m and 6~m, the bandwidth are almost the same, which implies that it is not affected by distance.
\subsubsection{Channel Capacity}
The Shannon channel capacity is also evaluated. According to \cite{rappaport1996wireless}, ${\mathcal C}=Bw \cdot log_2[1+P_r/(Bw\cdot N_{noise})]$, where $Bw$ is the bandwidth and $N_{noise}$ is the noise density. We set the transmission power as 10~dBm and consider two noise levels, namely, high noise -110~dBm/Hz and low noise -130~dBm/Hz. As it shown in Fig. \ref{fig:capacity}, M$^2$I can increase the communication range significantly. Besides the low path loss of M$^2$I, the narrower bandwidth also brings lower noise power.
\section{Conclusion}
Metamaterial has been introduced to enlarge magnetic induction's communication range and data rate. However, the existing works are based on ideal metamaterial model. In this paper, we propose a practical design of Metamaterial-enhanced Magnetic Induction (M$^2$I) communication by using a spherical coil array. The physical principles and geometric structure of this design are introduced. The relation between this practical M$^2$I and the ideal metamaterial based M$^2$I are discussed. Through the communication performance evaluation, we find it can significantly increase the channel capacity, which in turn extends the communication range greatly.


\section*{Appendix}
\subsection{Mutual Inductance}
The mutual inductance utilized in \eqref{equ:kirchhoff,equ:rearrange1,equ:effmu,equ:app1} is derived here. Without loss generality, in Fig. \ref{fig:direction}, the mutual inductance between coil 0 and coil 1 can be denoted as $M_{01}$, where coil 0 is transmitting coil and coil 1 is receiving coil. We can find the magnetic field intensity radiated by coil 0 by using \eqref{equ:antenna_field}. Let the magnetic field intensity at coil 1 be $H_0 {\hat h}$, where ${\hat h}$ is the magnetic field direction. Then the mutual inductance can be expressed as $M_{01}={\mu_0 H_0 \pi a_1^2 {\hat h}\cdot {\hat o_1}}/{I_0},$
where $a_1$ is the radius of coil 1, $I_0$ is the current in coil 0, and ${\hat o_1}$ is the orientation unit vector of coil 1.
\subsection{Proof of Proposition \ref{pro:IH}}
First, in \eqref{equ:kirchhoff} ${\bf Z}{\bf I}_0={\bf V}_0$. Assume that all the elements in ${\bf V}_0$ are the same, i.e., $V_1=V_2=\cdots=V_N=-j\omega\pi a^2 \mu_0 {\hat H_r}$. Since all the coils are uniformly distributed on the spherical shell, the structure is symmetrical. Moreover, as the coils have the same excitation voltages, they should have the same induced currents, i.e., $I_1=I_2=\cdots=I_N$. Equation \eqref{equ:kirchhoff} reduces to $I_1(Z_1+\sum_{i\neq 1}^{N}j\omega M_{1i})=V_1$. Thus, $I_1={\mathcal A} {\hat H_r}$, where ${\mathcal A}=\frac{-j\omega\pi a^2 \mu_0}{R_c+j\omega L_c+1/j\omega C+\sum_{i=1,i\neq n}^{N}j\omega M_{in}}$ is a constant. Since the coils are uniformly distributed, $\sum_{i=1,i\neq n}^{N}j\omega M_{in}$ are the same for all the coils no matter which $n$ we choose.

Then, we consider the source is a magnetic loop antenna and \eqref{equ:kirchhoff} is updated as ${{\bf Z}}{\hat {\bf I}}={\hat {\bf V}}$. According to \eqref{equ:voltage}, $V_n={\mathcal A} {\hat H_r}\cos\theta_n$. Hence ${\hat{\bf V}}={\bf K}{\bf V}_0$, where ${\hat {\bf K}}$ is a diagonal matrix and $diag({\hat K_n})=\cos\theta_n$. Thus, ${{\bf Z}}{\hat {\bf I}}={\bf K}{\bf V}_0$. By multiplying the inverse of ${\bf K}$, we can obtain ${\bf K}^{-1}{ {\bf Z}}{\hat {\bf I}}={\bf V}_0$. Also, ${{\bf Z}}$ can be regarded as diagonal matrix since mutual inductances are much smaller than coil impedances, i.e., $j\omega M<<(R_c+j\omega L_c-j/\omega C)$. Thus, ${{\bf Z}}{\bf K}^{-1}{\hat {\bf I}}={\bf V}_0$ and ${\bf I}_0={\bf K}^{-1}{\hat {\bf I}}$. Then, we can obtain ${\hat {\bf I}}={\bf K}{\bf I}_0$. As a result, ${\hat I}_n=\cos\theta_n I_n={\mathcal A} {\hat H_r}\cos\theta_n$.
\bibliographystyle{IEEEtran}
\bibliography{meta_shell_2}

\begin{thebibliography}{10}
\providecommand{\url}[1]{#1}
\csname url@samestyle\endcsname
\providecommand{\newblock}{\relax}
\providecommand{\bibinfo}[2]{#2}
\providecommand{\BIBentrySTDinterwordspacing}{\spaceskip=0pt\relax}
\providecommand{\BIBentryALTinterwordstretchfactor}{4}
\providecommand{\BIBentryALTinterwordspacing}{\spaceskip=\fontdimen2\font plus
\BIBentryALTinterwordstretchfactor\fontdimen3\font minus
  \fontdimen4\font\relax}
\providecommand{\BIBforeignlanguage}[2]{{%
\expandafter\ifx\csname l@#1\endcsname\relax
\typeout{** WARNING: IEEEtran.bst: No hyphenation pattern has been}%
\typeout{** loaded for the language `#1'. Using the pattern for}%
\typeout{** the default language instead.}%
\else
\language=\csname l@#1\endcsname
\fi
#2}}
\providecommand{\BIBdecl}{\relax}
\BIBdecl

\bibitem{Sun_MI_TAP_2010}
Z.~Sun and I.~F. Akyildiz, ``Magnetic induction communications for wireless
  underground sensor networks,'' \emph{IEEE Transactions on Antenna and
  Propagation}, vol.~58, no.~7, pp. 2426--2435, July 2010.

\bibitem{Guo_underwater}
H.~Guo, Z.~Sun, and P.~Wang, ``Channel modeling of mi underwater communication
  using tri-directional coil antenna,'' in \emph{IEEE Globecom 2015}, San
  Diego, USA, Dec 2015.

\bibitem{karlsson2004physical}
A.~Karlsson, ``Physical limitations of antennas in a lossy medium,''
  \emph{Antennas and Propagation, IEEE Transactions on}, vol.~52, no.~8, pp.
  2027--2033, 2004.

\bibitem{guo2015m2i}
H.~Guo, Z.~Sun, J.~Sun, and N.~Litchinitser, ``M2i: Channel modeling for
  metamaterial-enhanced magnetic induction communications,'' \emph{Antennas and
  Propagation, IEEE Transactions on}, vol.~63, no.~12, pp. 1--1, December 2015.

\bibitem{pendry1999magnetism}
J.~B. Pendry, A.~J. Holden, D.~Robbins, and W.~Stewart, ``Magnetism from
  conductors and enhanced nonlinear phenomena,'' \emph{Microwave Theory and
  Techniques, IEEE Transactions on}, vol.~47, no.~11, pp. 2075--2084, 1999.

\bibitem{Wang_WPT}
B.~Wang, W.~Yerazunis, and K.~H. Teo, ``Wireless power transfer: metamaterials
  and array of coupled resonators,'' \emph{Proceedings of the IEEE}, vol. 101,
  no.~6, 2013.

\bibitem{scarborough2012experimental}
C.~Scarborough, Z.~Jiang, D.~Werner, C.~Rivero-Baleine, and C.~Drake,
  ``Experimental demonstration of an isotropic metamaterial super lens with
  negative unity permeability at 8.5 mhz,'' \emph{Applied Physics Letters},
  vol. 101, no.~1, p. 014101, 2012.

\bibitem{xie2012proposal}
Y.~Xie, J.~Jiang, and S.~He, ``Proposal of cylindrical rolled-up metamaterial
  lenses for magnetic resonance imaging application and preliminary
  experimental demonstration,'' \emph{Progress In Electromagnetics Research},
  vol. 124, pp. 151--162, 2012.

\bibitem{schurig2006metamaterial}
D.~Schurig, J.~Mock, B.~Justice, S.~A. Cummer, J.~B. Pendry, A.~Starr, and
  D.~Smith, ``Metamaterial electromagnetic cloak at microwave frequencies,''
  \emph{Science}, vol. 314, no. 5801, pp. 977--980, 2006.

\bibitem{EMT_theory}
W.~Merrill, R.~Diaz, M.~Lore, M.~Squires, and N.~Alexopoulos, ``Effective
  medium theories for artificial materials composed of multiple sizes of
  spherical inclusions in a host continuum,'' \emph{IEEE Transactions on
  Antennas and Propagation}, vol.~47, no.~1, January 1999.

\bibitem{Ziolkowski_electric}
R.~W. Ziolkowski and A.~D. Kipple, ``Application of double negative materials
  to increase the power radiated by electrically small antennas,'' \emph{IEEE
  Transactiions on Antenna and Propagation}, vol.~51, no.~10, pp. 2626--2640,
  2003.

\bibitem{erentok2008metamaterial}
A.~Erentok and R.~W. Ziolkowski, ``Metamaterial-inspired efficient electrically
  small antennas,'' \emph{Antennas and Propagation, IEEE Transactions on},
  vol.~56, no.~3, pp. 691--707, 2008.

\bibitem{caloz2005electromagnetic}
C.~Caloz and T.~Itoh, \emph{Electromagnetic metamaterials: transmission line
  theory and microwave applications}.\hskip 1em plus 0.5em minus 0.4em\relax
  John Wiley \& Sons, 2005.

\bibitem{horn1984extended}
B.~K. Horn, ``Extended gaussian images,'' \emph{Proceedings of the IEEE},
  vol.~72, no.~12, pp. 1671--1686, 1984.

\bibitem{Balanis_a}
C.~A. Balanis, \emph{Antenna Theory}.\hskip 1em plus 0.5em minus 0.4em\relax
  John Wiley Publishing Company, 2005.

\bibitem{Guo_oil_2014}
H.~Guo and Z.~Sun, ``Channel and energy modeling for self-contained wireless
  sensor networks in oil reservoirs,'' \emph{IEEE Transactions on Wireless
  Communications}, vol.~13, no.~4, pp. 2258--2269, April 2014.

\bibitem{rappaport1996wireless}
T.~S. Rappaport \emph{et~al.}, \emph{Wireless communications: principles and
  practice}.\hskip 1em plus 0.5em minus 0.4em\relax prentice hall PTR New
  Jersey, 1996, vol.~2.

\end{thebibliography}

\end{document}